\documentclass[twocolumn,amsmath,amssymb,aps,pra]{revtex4-1}
\usepackage{graphicx,braket,amsthm,array,hyperref}
\pdfoutput=1
\hypersetup{pdfauthor={Matthew F. Pusey},pdftitle={Stabilizer notation for Spekkens' toy theory}}
\newcommand{\abs}[1]{\left\lvert #1 \right\rvert}
\newcommand{\gen}[1]{\left\langle #1 \right\rangle}
\newcommand{\I}{\ensuremath{{\mathcal{I}}}}
\newcommand{\X}{\ensuremath{{\mathcal{X}}}}
\newcommand{\Y}{\ensuremath{{\mathcal{Y}}}}
\newcommand{\Z}{\ensuremath{{\mathcal{Z}}}}
\newcommand{\sw}[4]{\ensuremath{#1\cdot#2 \leftrightarrow #3\cdot#4}}
\newtheorem{lemma}{Lemma}

\DeclareMathOperator{\tr}{Tr}
\DeclareMathOperator{\diag}{diag}
\begin{document}
\title{Stabilizer notation for Spekkens' toy theory}
\author{Matthew F. Pusey}
\email{m@physics.org}
\affiliation{QOLS, Blackett Laboratory, Imperial College London, Prince Consort Road, London SW7 2BW, United Kingdom}
\date{\today}
\begin{abstract}
  Spekkens has introduced a toy theory [Phys. Rev. A, \textbf{75}, 032110 (2007)] in order to argue for an epistemic view of quantum states. I describe a notation for the theory (excluding certain joint measurements) which makes its similarities and differences with the quantum mechanics of stabilizer states clear. Given an application of the qubit stabilizer formalism, it is often entirely straightforward to construct an analogous application of the notation to the toy theory. This assists calculations within the toy theory, for example of the number of possible states and transformations, and enables superpositions to be defined for composite systems.
\end{abstract}
\maketitle

What is the quantum state? The \emph{ontic view} holds that it is a property of the physical system. The \emph{epistemic view} is that it merely represents some agent's knowledge about the system. To support the latter view, Spekkens has constructed a toy theory \cite{spek} where the underlying physical systems are classical yet many quantum features are recovered through an epistemic restriction.

The states, transformations and measurements in the toy theory bear a striking resemblance to those described by the stabilizer formalism for qubits. That formalism describes a non-trivial subset of the quantum mechanical states, transformations and measurements of qubits in a much more compact manner than the normal Hilbert space formalism. A surprising consequence, known as the Gottesman-Knill theorem, is that a non-trivial subset of quantum mechanics can be efficiently simulated on a classical computer. The subset is rich enough to include many quantum phenomena, for instance entanglement, non-locality, quantum teleportation and dense coding.

The formalism was originally developed \cite{firststab} to study quantum error correcting codes, but has seen widespread use since, for example in the study of measurement based quantum computation \cite{mbqc}. I will review the relevant parts of the qubit stabilizer formalism during this paper. For the reader that has never encountered it before, a full introduction to that formalism can be found in Section 10.5 of \cite{nc}, and further useful details in \cite{caves,stabz2}.
\begin{table}[b!]
\begin{tabular}{|c||c|>{\centering}m{2cm}||m{1.3cm}|>{\centering}m{2cm}|}
\hline 
 & \multicolumn{2}{c||}{Quantum mechanics} & \multicolumn{2}{c|}{Toy theory}\tabularnewline
\hline 
 & Hilbert space & Stabilizer formalism & \centering Ontic space & Stabilizer notation\tabularnewline
\hline
\hline 
1 & $\begin{matrix}\Ket{\Phi^{+}}=\\\frac{1}{\sqrt{2}}(\Ket{00}+\Ket{11})\end{matrix}$ & $\begin{matrix}\langle X_{1}X_{2},\\Z_{1}Z_{2}\rangle\end{matrix}$ & \vspace{0.2em}\includegraphics[width=1.3cm]{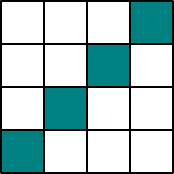} & $\begin{matrix}\langle\X_{1}\X_{2},\\\Z_{1}\Z_{2}\rangle\end{matrix}$\tabularnewline
\hline 
2 & $X = \begin{pmatrix} 0 & 1 \\ 1 & 0 \end{pmatrix}$ & $\begin{matrix}X_1 \to X_1,\\Z_1 \to -Z_1\end{matrix}$  & \vspace{0.2em}\includegraphics[width=1.3cm]{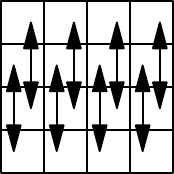} & $\begin{matrix}\X_1 \to \X_1,\\\Z_1 \to -\Z_1\end{matrix}$ \tabularnewline
\hline 
3 & $Z = \begin{pmatrix} 1 & 0 \\ 0 & -1 \end{pmatrix}$ & $\begin{matrix}X_1 \to -X_1,\\Z_1 \to Z_1\end{matrix}$  & \vspace{0.2em}\includegraphics[width=1.3cm]{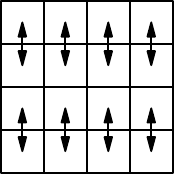} & $\begin{matrix}\X_1 \to -\X_1,\\\Z_1 \to \Z_1\end{matrix}$ \tabularnewline
  \hline
4 & $\begin{matrix}\ket{\Psi^-} =\\ \frac{1}{\sqrt 2}(\ket{01} - \ket{10})\end{matrix}$ & $\begin{matrix}\langle-X_1X_2,\\-Z_1Z_2\rangle\end{matrix}$  & \vspace{0.2em}\includegraphics[width=1.3cm]{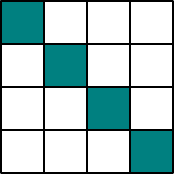} & $\begin{matrix}\langle-\X_1\X_2,\\-\Z_1\Z_2\rangle\end{matrix}$ \tabularnewline
\hline 
5 & $\begin{matrix}\{\ket{\Phi^+}, \ket{\Phi^-},\\ \ket{\Psi^+}, \ket{\Psi^-}\}\end{matrix}$ & $X_1X_2$ and $Z_1Z_2$ & \vspace{0.2em}\includegraphics[width=1.3cm]{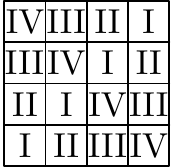} & $\X_1\X_2$ and $\Z_1\Z_2$ \tabularnewline
\hline
\end{tabular}
\caption{Dense coding \protect\cite{dense} in two theories (the quantum version closely follows Section 2.3 of \cite{nc}), each in two notations. For the benefit of those readers that are familiar with \cite{spek}, the column labelled ``ontic space'' reproduces the relevant diagrams from that paper. The remaining readers may ignore that column. 
Initially the joint state of two systems indicated in row 1 is prepared. Alice is given the first system and Bob the second. In isolation, each system can only be used to transmit a single bit of classical information, yet the dense coding protocol allows two bits $(b_1, b_2)$ to be sent from Alice to Bob with the transfer of only one system.  If $b_1 = 1$ then Alice performs the operation indicated in row 2 on her system. Similarly, if $b_2 = 1$ she performs the operation indicated in row 3. For example, if $(b_1, b_2) = (1,1)$ she performs both operations and the joint state of the system is then as shown in row 4. Finally, Alice sends her system to Bob, who performs the joint measurement indicated in row 5 to recover two bits of information.}
\label{densecoding}
\end{table}

The purpose of this paper is to provide a new notation for the toy theory which explains the similarities with the qubit stabilizer formalism, whilst also pinning down exactly how the predictions of the toy theory differ from those of quantum mechanics. A sneak preview of the notation is given by Table~\ref{densecoding}, showing how similar it is to the qubit stabilizer formalism. The key difference is that whilst for qubits $XZ = -iY$, in the toy theory $\X\Z = \Y$.

The fact that the toy theory can be described using a notation so similar to the qubit stabilizer formalism may itself be considered further evidence for the epistemic view of quantum states.

The notation is also useful for carrying out calculations in the theory, as shown by the examples in Section~\ref{examples}. This is primarily for two reasons. Firstly the notation is much more compact: the description of the key objects of the theory (pure epistemic states and reversible transformations) using the stabilizer notation scales polynomially in the number of systems. A direct description, as used in \cite{spek}, scales exponentially. Secondly, there is widespread experience with, and extensive literature on, the qubit stabilizer formalism. Much of this can be applied to the toy theory thanks to the notation provided here.

The use of the notation to describe states, transformations, and measurements is described in Sections~\ref{statessec}, \ref{transsec} and \ref{meassec} respectively. Convex combinations and coherent superpositions are discussed in Section~\ref{mixsec}.

Spekkens has already outlined a new phase-space based formulation of the toy theory \cite{spektalk} which is closely related to this notation as shown in Appendix~\ref{newformulation}. The theory has also been reformulated using category theory \cite{bob}, but since that is very different to the standard qubit stabilizer formalism it does not facilitate calculations and comparisons in the same way as the notation provided here.

\section{States}\label{statessec}
\subsection{Qubits}
Denote the 2 by 2 identity matrix by $I$ and let
\begin{equation}
  X = \begin{pmatrix} 0 & 1 \\ 1 & 0 \end{pmatrix},\ 
  Y = \begin{pmatrix} 0 & -i \\ i & 0 \end{pmatrix},\ 
  Z = \begin{pmatrix} 1 & 0 \\ 0 & -1 \end{pmatrix}.
\end{equation}
Define $P_n$, the \emph{Pauli group on $n$ qubits}, as the $2^n$ by $2^n$ matrices of the form $\alpha p_1 \otimes \dotsm \otimes p_n$ for some $\alpha \in \{1,-1,i,-i\}$ and $p_k \in \{I,X,Y,Z\}$. Call the Hermitian elements (i.e. those with $\alpha \in \{1, -1\}$) \emph{Pauli observables}. Define $X_k$ as $X$ acting on the $k$-th qubit, i.e. $I^{\otimes(k-1)}\otimes X \otimes I^{\otimes(n-k)}$. Similarly for $Y_k$ and $Z_k$. $P_n$ is generated by the $X_k$ and the $Z_k$ along with $iI^{\otimes n}$. 

An element $g$ of the Pauli group, ignoring its phase $\alpha$, can be written as a \emph{check vector} $r(g)$. This is a vector of $2n$ bits, where the first $n$ bits give the locations of $X$s, whilst the second $n$ bits give the locations of $Z$s, a $Y$ being indicates by 1s in both positions. For example, $r(X \otimes Y \otimes Z) = (1,1,0,0,1,1)$, $r(-Z \otimes I) = (0,0,1,0)$. A useful property of check vectors is that $r(gh) = r(g) \oplus r(h)$, where $\oplus$ indicates addition modulo $2$ (making $r$ a group homomorphism from $P_n$ to $\left( \mathbb{Z}_2 \right)^{2n}$).

Let $S$ be a subgroup of $P_n$ that does not contain $-I^{\otimes n}$ (a \emph{qubit stabilizer subgroup}). It is conventional to associate $S$ with the subspace $V_S$ of pure states $\ket{\psi}$ satisfying $g\ket{\psi} = \ket{\psi}$ for all $g \in S$. The projector onto this subspace is \cite{caves}
\begin{equation} P_S = \frac{1}{\abs{S}} \sum_{g \in S} g. \label{projector}\end{equation}

  For comparison with what follows it is useful instead to associate $S$ with the quantum state $\rho_S = \abs{S} 2^{-n} P_S$. This is a pure state if and only if $\abs{S} = 2^n$. Otherwise it is a uniform mixture of $\frac{2^n}{\abs{S}}$ pure states (which form a basis for $V_S$). States of this form are also considered in \cite{ag} and \cite{normalforms}.

  For any Pauli observable $g$
  \begin{equation}
    \tr(g\rho_S) = \begin{cases} 1 & g \in S \\ -1 & -g \in S \\ 0 & \text{otherwise}\end{cases}.
  \end{equation}
  In epistemic language we could say that $\rho_S$ represents complete knowledge about the Pauli observables $g$ with $\pm g \in S$ and zero knowledge about the rest.

\subsection{Toy theory}

In Spekkens toy theory, the simplest systems, called \emph{elementary systems}, are always in one of the four possible \emph{ontic states}. The ontic state is a ``hidden variable'' that completely describes the physical state of affairs of the system.  In the stabilizer notation these ontic states are identified with the vectors $\vec e_1 = (1,0,0,0)$, $\vec e_2 = (0,1,0,0)$, $\vec e_3 = (0,0,1,0)$ and $\vec e_4 = (0,0,0,1)$.

A \emph{composite system} is composed of elementary systems, and its ontic state is specified by specifying the ontic state of each elementary system. For example $2 \cdot 4$ means that the first system is in ontic state $2$ and the second is in state $4$.  In the stabilizer notation this is identified with the tensor product $\vec e_2 \otimes \vec e_4$, whilst $1 \cdot 3 \cdot 2$ is identified with $\vec e_1 \otimes \vec e_3 \otimes \vec e_2$, and so on.  In this way ontic state for $n$ elementary systems will therefore correspond to a $4^n$-dimensional vector $\vec v$ with a single component taking the value 1 and the rest zero.

An \emph{epistemic state} describes an agent's knowledge about a system. An agent that assigns the epistemic state $1 \vee 2$ to an elementary systems knows that its ontic state is either 1 or 2. Crucial to Spekkens' argument is the observation that the epistemic states resemble quantum states, whilst the ontic states do not.

Notice that the knowledge is incomplete: the agent does not know the exact ontic state. This is required by the theory's ``knowledge balance principle'' \cite{spek}. Detailed discussion of this principle is relegated to Appendix~\ref{proofs}, where it is shown that the epistemic states satisfying it are exactly those that can represented as follows:

Denote the 4 by 4 identity matrix by $\I = \diag(1,1,1,1)$ and let
\begin{align}
  \X = \diag(1,-1,1,-1),\\ \Y = \diag(1,-1,-1,1),\\ \Z = \diag(1,1,-1,-1).
\end{align}

These correspond to the possible measurements of an elementary system, where if the ontic state is $k$ then a measurement of $g \in \{\I, \X, \Y, \Z\}$ will return $v \in \{1, -1\}$ according to the eigenvalue equation $g\vec e_k = v\vec e_k$.

Define $G_n$ as the $4^n$ by $4^n$ matrices of the form $\alpha g_1 \otimes \dotsb \otimes g_n$ for some $\alpha \in \{1,-1\}$ and $g_k \in \{ \I, \X, \Y, \Z \}$. This group plays the same role as the Pauli group does in the qubit case. Call the elements \emph{toy observables}. Denote $\X_k = \I^{\otimes (k-1)} \otimes \X \otimes \I^{\otimes (n-k)}$, and similarly for $\Y_k$ and $\Z_k$. $G_n$ is generated by $-\I^{\otimes n}$ together with $\X_k$ and $\Z_k$ for all $k \in \{1, \dotsc, n\}$.

Let $m:G_n \to P_n$ be the mapping suggested by our notation, so that $m(-\X \otimes \Z) = -X \otimes Z$ and so on. Abuse terminology by saying $g, h \in G_n$ \emph{commute (anticommute)} if $m(g)$ and $m(h)$ commute (anticommute). Define the \emph{check vector} of $g \in G_n$ to be $r(m(g))$, the check vector of $m(g)$. As with the qubit case, we have that multiplication in $G_n$ corresponds to addition of check vectors modulo two. (Formally, $r\circ m : G_n \to \left( \mathbb{Z}_2 \right)^{2n}$ is a group homomorphism, even though $m: G_n \to P_n$ is not. Furthermore, appending a ``phase bit'' $\frac12(1+\alpha)$ to the check vector gives a group isomorphism between $G_n$ and $(\mathbb{Z}_2)^{(2n+1)}$.)

Let $S$ be a subgroup of $G_n$ that does not contain $-\I^{\otimes n}$ and for which every $g, h \in S$ commute (a \emph{toy stabilizer subgroup}). $S$ represents an epistemic state for $n$ elementary systems, namely the knowledge that the ontic state $\vec v$ satisfies $g\vec v = \vec v$ for all $g \in S$. This is a subgroup because if $g\vec v = \vec v$ and $h \vec v = \vec v$ then $gh\vec v = \vec v$ also. The projector $P_S$ onto the ontic states compatible with $S$ is of the form \eqref{projector}. 

\begin{table}
  \centering
  \begin{tabular}{| >{\centering}m{2.5cm} | >{\centering}m{2.6cm} | c |}
    \hline
    Picture & Ontic states & Toy stabilizers \\
    \hline
    \vspace{0.5em}
    \includegraphics{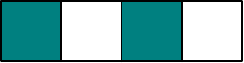}       & $1 \vee 3$               & $\gen{\X}$  \\
    \includegraphics{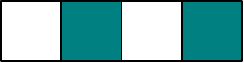}       & $2 \vee 4$               & $\gen{-\X}$ \\
    \includegraphics{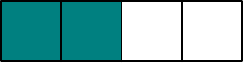}       & $1 \vee 2$               & $\gen{\Z}$  \\
    \includegraphics{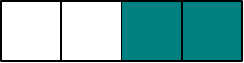}       & $3 \vee 4$               & $\gen{-\Z}$ \\
    \includegraphics{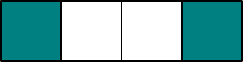}       & $1 \vee 4$               & $\gen{\Y}$  \\
    \includegraphics{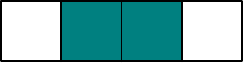}       & $2 \vee 3$               & $\gen{-\Y}$ \\
    \includegraphics{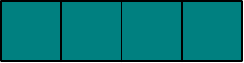} & $1 \vee 2 \vee 3 \vee 4$ & $\gen{}$    \\
    \hline
  \end{tabular}
  \caption{All of the toy stabilizer subgroups for an elementary system. The filled boxes in the first column correspond to the possible ontic states, as in \cite{spek}. The symbol $\vee$ should be read as ``or''.}
  \label{elstates}
\end{table}

\begin{table}
  \centering
  \begin{tabular}{| >{\centering}m{1.95cm} | >{\centering}m{2.5cm} | c |}
    \hline
    Picture & Ontic states & Toy stabilizers \\
    \hline
    \vspace{0.5em}\includegraphics{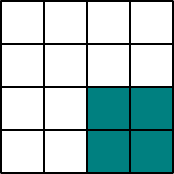}                                       & $(1\cdot 3) \vee (1\cdot 4) \vee (2\cdot 3) \vee (2\cdot 4)$               & $\gen{\Z_1, -\Z_2}$ \\
    \includegraphics{figs/11v22v33v44}                                                 & $(1\cdot 1) \vee (2\cdot 2) \vee (3\cdot 3) \vee (4\cdot 4)$               & $\gen{\Z_1\Z_2, \X_1\X_2}$ \\
    \includegraphics{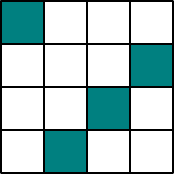}                                                 & $(1\cdot 2) \vee (2\cdot 3) \vee (3\cdot 4) \vee (4\cdot 1)$               & $\gen{-\Z_1\Y_2, -\X_1\X_2}$ \\
    \includegraphics{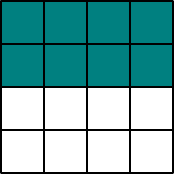}                                                      & $(3\vee4) \cdot (1\vee2\vee3\vee4)$                                        & $\gen{-\Z_1}$ \\
    \includegraphics{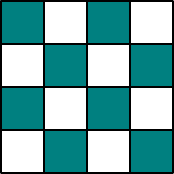}                                                      & $[(1\vee 3)\cdot (2\vee4)]\vee [(2\vee 4)\cdot (1\vee 3)]$                 & $\gen{-\X_1\X_2}$ \\
    \includegraphics{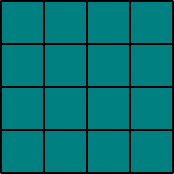}                                             & $(1\vee2\vee3\vee4)\cdot(1\vee2\vee3\vee4)$                                & $\gen{}$ \\
    \includegraphics{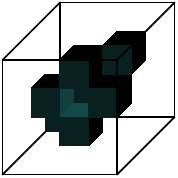}                                                         & $(1\cdot 1\cdot 1)\vee (1\cdot 2\cdot 2)\vee (2\cdot 1\cdot 2)\vee (2\cdot
    2\cdot 1)\vee (3\cdot 3\cdot 3)\vee (3\cdot 4\cdot 4)\vee (4\cdot 3\cdot 4)\vee(4\cdot 4\cdot 3)$ & $\gen{\X_1\X_2\X_3,\Z_1\Z_2,\Z_2\Z_3}$ \\
    \hline
  \end{tabular}
  \caption{Some toy stabilizer subgroups for composite systems. They are each composed of two elementary systems, except for the last which is composed of three. The pictures use the same conventions as in \cite{spek}, briefly: each axis corresponds to an elementary system, and the possible ontic states are filled. In the two-system cases each row is a state of the first system and each column a state of the second, with the ontic state $1\cdot1$ in the bottom-left corner. In the three-system case the ``depth'' gives the state of the third system, with the ontic state $1 \cdot 1 \cdot 1$ is in the bottom-left-back corner.}
  \label{compstates}
\end{table}

Notice that whilst in the qubit case $-I^{\otimes n} \notin S$ automatically ensures the elements of $S$ commute, in the toy theory the commuting requirement is added ``by hand''.

The only element of $S$ with non-zero trace is $\I^{\otimes n}$ and so $\tr{P_S}$, which is the number of ontic states compatible with $S$, is $4^n / \abs{S}$. Since we assume a uniform distribution over the possible ontic states, $\rho_s = \abs{S}4^{-n}P_S$ gives a diagonal matrix of the probabilities for each ontic state. 

Some examples of toy stabilizer subgroups are shown in Tables~\ref{elstates} and \ref{compstates}. They are all reminiscent of qubit stabilizer states, and the following Lemma shows why.
\begin{lemma}[proven in Appendix~\ref{proofs}]
  $g_1, g_2, \dotsc, g_l \in G_n$ are independent generators of a toy stabilizer subgroup
  if and only if
  $m(g_1), m(g_2), \dotsc, m(g_l)$ are independent generators of a qubit stabilizer subgroup.
  \label{samegen}
\end{lemma}

It is crucial to note that although valid lists of independent generators are the same in both theories, the resultant subgroups are in general different. For example $\gen{\X_1\X_2, \Y_1\Y_2} = \{ \I^{\otimes 2}, \X_1\X_2, \Y_1\Y_2, \Z_1\Z_2\}$ whilst $\gen{X_1X_2, Y_1Y_2} = \{ I^{\otimes 2}, X_1X_2, Y_1Y_2, -Z_1Z_2\}$.

It is easy to see that if $g_1, \cdots, g_l$ are independent generators for $S$ then $\abs{S} = 2^l$. The maximum number of independent generators for a qubit stabilizer subgroup is $n$ \cite{nc}, and the above lemma means this is also true for toy stabilizer subgroups. Hence an $S$ with $\abs{S} = 2^n$ represents a state of maximal knowledge, or a \emph{pure state}.

Two epistemic states are called \emph{disjoint} if no ontic state is compatible with both. In this notation there is a useful criterion for disjoint states, which is identical to the criterion for orthogonality in the qubit stabilizer case.
\begin{lemma}
  A pair of toy stabilizer subgroups $S, T \subset G_n$ represent disjoint epistemic states if and only if there exists some $g \in S$ with $-g \in T$.
\end{lemma}
\begin{proof}
  $S$ and $T$ represent disjoint states if and only if $0 = P_SP_T = P_{\gen{S \cup T}}$. This holds if and only if $-\I^{\otimes n} \in \gen{S \cup T}$. But since $-\I^\otimes{n} \notin S, T$, this holds if and only if there exists $g \in S$ with $-g \in T$.
\end{proof}

Finally, note that the theory of qubits restricted to the computational basis is just the classical probability theory of bits. It is occasionally helpful to view the above notation as being based on the encoding of the ontic level of the toy theory (two classical bits per elementary system) in the computational basis states of qubits. Since $\X = I \otimes Z$, $\Y = Z \otimes Z$ and $\Z = Z \otimes I$, the toy observables are then Hermitian observables on those qubits. This device should not be taken too seriously, for example (as discussed in Section~\ref{meassec}) a measurement of $\X$ can disturb the value of $\Z$ in the toy theory, whereas the observables $I \otimes Z$ and $Z \otimes I$ are compatible in quantum mechanics.

\section{Transformations}\label{transsec}
\subsection{Qubits}
  Define $C_n$, the \emph{Clifford group on n qubits} as the $2^n$ by $2^n$ unitaries $U$ satisfying $UgU^\dagger \in P_n$ for all $g \in P_n$. These take stabilizer states to stabilizer states, since $U \rho_S U^\dagger = \rho_{T}$ where $T = \{ UgU^\dagger | g \in S \}$ is another stabilizer subgroup.
  
  Since $UghU^\dagger = UgU^\dagger UhU^\dagger$ we can specify the action of $U$ by its action on the generators of $P_n$. Since $UiI^{\otimes n}U^\dagger = iI^{\otimes n}$ it suffices to specify the action on the $X_k$ and $Z_k$. Note that $[UgU^\dagger, UhU^\dagger] = U[g,h]U^\dagger$. Hence if $UX_kU^\dagger = g_k$ and $UZ_kU^\dagger = h_k$ we will have that the $g_k$ commute, the $h_k$ commute, and $g_k$ commutes with $h_l$ if and only if $k \neq l$. Since $U$ is reversible we must have that the $h_k$ and $g_k$, together with $iI^{\otimes n}$, generate $P_n$. 
  
  Such $h_k, g_k$ are described in \cite{caves} as \emph{canonical generator sets}. It is then shown that, conversely, for any given canonical generator set there is a Clifford unitary that maps the $X_k$ and $Z_k$ to it. That unitary can be generated from the single qubit Clifford unitaries plus the controlled-NOT gate, which sends $X_1 \to X_1X_2, X_2 \to X_2, Z_1 \to Z_1, Z_2 \to Z_1Z_2$. Furthermore the single qubit Clifford unitaries can themselves be generated from the Hadamard gate, which sends $X \to Z, Z \to X$, and the phase gate, which sends $X \to Y, Z \to Z$.

\subsection{Toy theory}
In the toy theory reversible transformations are permutations of the ontic states that take any valid epistemic state to another valid epistemic state. In this section I show that the group of reversible transformations in the toy theory has a very similar structure to the Clifford group. The transformations act on the $\X_k$ and $\Z_k$ in the same way as the qubit case. Furthermore the transformations can be generated in a similar fashion to the Clifford group.

We can write a reversible transformation on $n$ elementary systems as a $4^n$ by $4^n$ permutation matrix $U$.  For example, the permutation matrix for the permutation $1 \to 2 \to 3 \to 1$ on an elementary system is
\begin{equation}
  U = \begin{pmatrix}
    0 & 0 & 1 & 0 \\
    1 & 0 & 0 & 0 \\
    0 & 1 & 0 & 0 \\
    0 & 0 & 0 & 1
  \end{pmatrix}
\end{equation}
If $\vec v$ is the ontic state before the transformation then $U\vec v$ is the state afterwards. Hence if we know that the ontic state is in the support of a projector $P_S$ before then our new knowledge is exactly that the ontic state is in the support of the projector $U P_S U^T$.

Before the transformation the epistemic state may have been $\gen{g}$ for any $g \in G_n$. Since $P_{\gen{g}} = \frac12 (\I^{\otimes n} + g)$, for the new epistemic state to be valid (and hence represented by a toy stabilizer subgroup) we must have $UgU^T \in G_n$. Compare this to the definition of the Clifford group. To ensure the elements of the new subgroup commute we must also require that $UgU^T$ and $UhU^T$ commute if and only if $g$ and $h$ commute. 

The theory now develops much like the qubit case. Since $UghU^T = UgU^TUhU^T$ and $U(-\I^{\otimes n})U^T = -\I^{\otimes n}$ we can specify the action of $U$ by its action on the $\X_k$ and $\Z_k$. The resulting elements must have the same commutation structure and, along with $-\I^{\otimes n}$, generate $G_n$. Call such elements canonical generating sets, and note that applying $m$ gives a qubit canonical generator set.

For an elementary system, it is shown in Table~\ref{eltrans} that all of the $4! = 24$ permutations of the ontic states are valid, and that there is a permutation that sends the $\X$ and $\Z$ to any canonical generator set.

\begin{table}
  \centering
  \begin{tabular}{|c|c|c|}
    \hline
    Permutation & Effect on toy stabilizers \\
    \hline
(1)(2)(3)(4) & $\X \to \X$, $\Z \to \Z$   \\
(1)(2)(43)   & $\X \to \Y$, $\Z \to \Z$   \\
(1)(32)(4)   & $\X \to \Z$, $\Z \to \X$   \\
(1)(342)     & $\X \to \Y$, $\Z \to \X$   \\
(1)(432)     & $\X \to \Z$, $\Z \to \Y$   \\
(1)(42)(3)   & $\X \to \X$, $\Z \to \Y$   \\
(21)(3)(4)   & $\X \to -\Y$, $\Z \to \Z$  \\
(21)(43)     & $\X \to -\X$, $\Z \to \Z$  \\
(231)(4)     & $\X \to \Z$, $\Z \to -\Y$  \\
(2341)       & $\X \to -\X$, $\Z \to -\Y$ \\
(2431)       & $\X \to \Z$, $\Z \to -\X$  \\
(241)(3)     & $\X \to -\Y$, $\Z \to -\X$ \\
(321)(4)     & $\X \to -\Y$, $\Z \to \X$  \\
(3421)       & $\X \to -\Z$, $\Z \to \X$  \\
(31)(2)(4)   & $\X \to \X$, $\Z \to -\Y$  \\
(341)(2)     & $\X \to -\Z$, $\Z \to -\Y$ \\
(31)(42)     & $\X \to \X$, $\Z \to -\Z$  \\
(3241)       & $\X \to -\Y$, $\Z \to -\Z$ \\
(4321)       & $\X \to -\X$, $\Z \to \Y$  \\
(421)(3)     & $\X \to -\Z$, $\Z \to \Y$  \\
(431)(2)     & $\X \to \Y$, $\Z \to -\X$  \\
(41)(2)(3)   & $\X \to -\Z$, $\Z \to -\X$ \\
(4231)       & $\X \to \Y$, $\Z \to -\Z$  \\
(41)(32)     & $\X \to -\X$, $\Z \to -\Z$ \\
\hline
  \end{tabular}
  \caption{The valid reversible transformations for an elementary system. The first column shows the permutations to the ontic states in cycle notation, for example (342) indicates that $3 \to 4 \to 2 \to 3$. The second shows the result of $UgU^T$ on two non-trivial generators of $G_1$.}
  \label{eltrans}
\end{table}

It is well known that the group of permutations on $n$ elements can be generated by a transposition and an $n$-cycle. Hence any elementary transformation can be written as a sequence of, for example $3 \leftrightarrow 2$ and $1 \to 4 \to 2 \to 3 \to 1$ transformations. The first is reminiscent of a Hadamard gate in that it maps $\X \to \Z$ and $\Z \to \X$, although it maps $\Y \to \Y$ whereas a Hadamard maps $Y \to -Y$. The second is reminiscent of a phase gate in that it maps $\X \to \Y$ and $\Y \to -\X$, although it maps $\Z \to -\Z$ whereas a phase gate maps $Z \to Z$.

\begin{table}
  \centering
  \begin{tabular}{| >{\centering}m{2.5cm} | >{\centering}m{3.8cm} |}
    \hline
    Permutation & Effect on toy stabilizers \tabularnewline
    \hline
    \vspace{0.2em}\includegraphics[width=2.5cm]{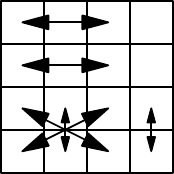} &
      $\X_1 \to -\Y_1$, $\X_2 \to \X_2$, $\Z_1 \to \Z_1$, $\Z_2 \to -\Y_2$\tabularnewline
    \includegraphics[width=2.5cm]{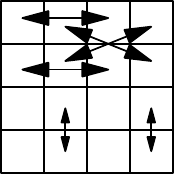} &
    $\X_1 \to \X_1 \X_2$, $\X_2 \to \X_2$, $\Z_1 \to \Z_1$, $\Z_2 \to \Z_1\Z_2$\tabularnewline
\hline
  \end{tabular}
  \caption{Two reversible transformations for pairs of elementary systems. The first column shows the permutations to the ontic states (laid out in a grid using the same conventions as \cite{spek} and Table~\ref{compstates}). The second shows the action on the non-trivial generators of $G_2$. The first transformation acts on each system separately. The second transformation is analogous to a controlled-NOT.}
  \label{comptrans}
\end{table}

The argument in \cite{caves} now gives that there is a permutation on a composite system that sends the $\X_k, \Z_k$ to any given canonical generator set, and that it can be generated using the elementary system transformations and the controlled-NOT gate shown in Table~\ref{comptrans}.

One difference between the transformations in the two theories has already been noted: transformations that act on $\X_k$ and $\Z_k$ in the same way in both theories do not necessarily act in the same way on other group elements, for example the $\Y_k$. Another difference is that whilst commutation structure is automatically preserved by any unitary transformation on qubits, this requirement must be added ``by hand'' in the toy theory.  For example, the matrix $U$ for the two-system permutation
\begin{equation}
  \sw1231, \sw1433, \sw2241, \sw2443
\end{equation}
satisfies the requirement that $UgU^\dagger \in G_2$ for all $g \in G_2$, but has $U\X_1U^T = \X_1$ and $U\X_2U^T = \Z_1$ and so is not a valid transformation.

\section{Measurements}\label{meassec}
\subsection{Qubits}
  Suppose a Pauli observable $g$ is measured on a state described by $S$. If $\pm g \in S$ then $\pm 1$ is returned and the state is unchanged.
  
If $\pm g \notin S$, then the result $v \in \{1, -1\}$ will be completely random. To find the state after the measurement, first write down a set of independent generators for $S$ with at most one element $h$ that anticommutes with $g$. (Any other elements that anticommute can be multiplied by $h$.) Add $vg$ to the list and remove $h$ (if present) to obtain a list of independent generators for the new state.

  \subsection{Toy theory}
A measurement in the toy theory is specified by partitioning the ontic state into valid epistemic states. The result of the measurement is simply whichever member of the partition the ontic state lies in. In order to ensure that the new epistemic state is valid, the measurement must disturb the ontic state. 

In the toy stabilizer notation we can describe measurements of toy observables $g$, i.e. the partitioning of the ontic states into $\gen{g}$ and $\gen{-g}$. In this section I will show that that measurements of toy observables behave in exactly the same way as Pauli observables in the qubit case, but that not every measurement in the toy theory can be described using toy observables.

Let the epistemic state before the measurement be described by a toy stabilizer subgroup $S$. If $g \in S$ then all the possible ontic states lie in the $+1$ eigenspace of $g$, so we are certain to get the $+1$ outcome. Similarly if $-g \in S$ we are certain to get the $-1$ outcome. Notice that the expected value of $g$ is $\tr(g\rho_S)$. If $g, -g \notin S$ then this gives $0$. Hence the measurement of such a toy observable must give either outcome with equal probability.

In order to ensure that the epistemic state after the measurement is valid, there must be a random disturbance to the value of any anticommuting toy observables. Let us assume this is the only disturbance - all commuting observables are unaffected.

If $\pm g$ was already in the subgroup then there is no change to the epistemic state. Otherwise, if the outcome of a measurement is $\pm 1$ the new toy stabilizer subgroup will be generated by $\pm g$ and the old stabilizer subgroup without any elements that anticommute with $g$. A systematic way to update the toy stabilizer subgroup is to follow the qubit stabilizer procedure: write a list of independent generators for the old state with at most one element that anticommutes with $g$, delete that element if present, and add $\pm g$.

An example of a valid measurement with four outcomes is the partitioning of the ontic states into those associated with $\gen{\X_1\X_2, \Z_1\Z_2}$, $\gen{-\X_1\X_2, \Z_1\Z_2}$, $\gen{\Z_1, -\Z_2}$ and $\gen{-\Z_1, \Z_2}$. This measurement can be handled in the stabilizer notation by first considering a measurement of $\Z_1\Z_2$. The $+1$ outcome means that the ontic state must lie in one of the first two sets and we can measure $\X_1\X_2$ to determine which. Similarly the $-1$ outcome narrows it down to the last two sets and we then measure $\Z_1$. 

Can every valid measurement be described using sequences of toy observables in this way? Somewhat surprisingly, the answer is no.

\begin{lemma}[proven in Appendix~\ref{proofs}]
    Any valid measurement on one or two elementary systems is equivalent to measuring a sequence of toy observables.
    \label{onetwomeas}
  \end{lemma}

  \begin{lemma}
    Not every valid measurement on three or more elementary systems is equivalent to measuring a sequence of toy observables.
    \label{threemeas}
  \end{lemma}
  \begin{proof}
    Consider the measurement on three or more elementary systems with eight outcomes
    \begin{equation}
      \begin{matrix}
      \gen{\Z_1, \Z_2, \X_3},\  \gen{-\Z_1, \X_2, \Z_3},\  \gen{\X_1, -\Z_2, -\Z_3},\\
      \gen{\Z_1, \Z_2, -\X_3},\  \gen{-\Z_1, -\X_2, \Z_3},\  \gen{-\X_1, -\Z_2, -\Z_3},\\
      \gen{\Z_1, -\Z_2, \Z_3},\  \gen{-\Z_1, \Z_2, -\Z_3}.
    \end{matrix}
    \label{nonstabmeas}
    \end{equation}
  There is no non-trivial toy observable $g$ with $\pm g$ in all the outcome subgroups (indeed there is no such $g$ for the first three outcomes alone). Therefore there is no $g$ that can be the first in the sequence of toy observable measurements.
  \end{proof}

  Since the only valid measurements on elementary systems are toy observables (Lemma~\ref{onetwomeas}), the above set of eight uncorrelated (product) states cannot be distinguished using local measurements alone. The same situation arises in quantum mechanics, where it is known as ``non-locality without entanglement'' \cite{qnwe}. Indeed equation 43 in \cite{qnwe} is very similar to \eqref{nonstabmeas}. As noted in \cite{spek}, the presence of this effect in the toy theory, which is local by construction, may indicate that ``non-locality without entanglement'' is a misnomer.
  
  A similar structure also arises in so-called ``boxworld'' \cite{boxworld}: the measurement in their proof of Theorem 3 is also reminiscent of \eqref{nonstabmeas}.

\section{Mixtures and superpositions}\label{mixsec}
\subsection{Qubits} 
I have been unable to find procedures for constructing mixtures (convex combinations) and superpositions of qubit stabilizer states in the literature, so I briefly develop the theory here.

Let $S$ and $S'$ be two stabilizer subgroups on $n$ qubits. When $g \in S$ if and only if $\pm g \in S'$, call $S'$ a \emph{rephasing} of $S$. In that case, either $S = S'$, or they represent orthogonal states and can be written $S = \gen{g_1, \dotsc, g_{l-1}, g_l}$, $S' = \gen{g_1, \dotsc, g_{l-1}, -g_l}$. For example, $S' = \gen{Z_1,-Z_2}$ is a rephasing of $S = \gen{-Z_1,Z_2}$. For each list of generators we then multiply first by the second to obtain $S' = \gen{-Z_1Z_2, -Z_2}$, $S = \gen{-Z_1Z_2, Z_2}$.

Consider the mixture $\frac12(\rho_S + \rho_{S'})$. It is easy to check that the result is equal to to $\rho_T$ for some stabilizer subgroup $T$ if and only if $S'$ is a rephasing of $S$. In that case $T = S \cap S'$ and $P_T \propto P_S + P{_S'}$. For example, with $S$ and $S'$ as above we would have $T = \gen{-Z_1Z_2}$.

The rephasing condition also applies to coherent superpositions, provided we restrict attention to orthogonal states. This excludes somewhat irregular cases such as $\ket{00} - \ket{++}$, which happens to be a stabilizer state.
\begin{lemma}[proven in Appendix~\ref{proofs}]
  \label{qubitsuper}
  Let $S, S', T$ be stabilizer subgroups of size $2^n$. Let $S$ and $S'$ stabilize orthogonal states $\ket{\psi}$ and $\ket{\psi'}$ respectively. Then there exists a $\theta \in \mathbb{R}$ such that $T$ stabilizes $\frac1{\sqrt2}(\ket{\psi} + e^{i\theta}\ket{\psi'})$ if and only if we can write $S = \gen{g_1, \dotsc, g_{n-1}, g_n}$, $S' = \gen{g_1, \dotsc, g_{n-1}, -g_n}$ and $T = \gen{g_1, \dotsc, g_{n-1}, h}$ for some $h \notin S \cup S'$.
\end{lemma}

There are always four such superpositions:
\begin{lemma}[proven in Appendix~\ref{proofs}]
  \label{supercount}
  Let $S = \gen{g_1, \dotsc, g_{n-1}, g_n}$ and $S' = \gen{g_1, \dotsc, g_{n-1}, -g_n}$. Then there are exactly four distinct stabilizer subgroups of the form $\gen{g_1, \dotsc, g_{n-1}, h}$ with $h \notin S \cup S'$.
\end{lemma}
For example, take $S = \gen{Z_1, Z_2}$ and $S' = \gen{-Z_1, -Z_2}$, which correspond to $\Ket{00}$ and $\Ket{11}$ respectively. Then we obtain four distinct superpositions of the form $T = \gen{Z_1Z_2, h}$ by setting $h$ to $\pm X_1X_2$ or $\pm X_1Y_2$, corresponding to $\frac1{\sqrt 2}(\ket{00} \pm \ket{11})$ or $\frac1{\sqrt 2}(\ket{00} \pm i\ket{11})$ respectively.

\subsection{Toy theory}
Define rephasing for toy stabilizer subgroups in the same way. Convex combinations then work in exactly the same way as the qubit case. If $S$ and $S'$ represent disjoint states then $P_{S} + P_{S'}$ is a projector onto the union of the epistemic states compatible with $S$ and $S'$, which is how convex combinations are defined in \cite{spek}. Defining convex combinations for any rephasing also allows $S = S'$, which is a trivial extension.

Since the toy theory does not have a direct analogue of Hilbert space there is no inherent definition of a coherent superposition. Instead we simply define it by analogy to the qubit case. Suppose $S$ and $S'$ are two toy stabilizer subgroups for pure states, and $S$ is a rephasing of $S'$. Then either $S = S'$ (in which case we can trivially define $S$ to be a coherent superposition of $S$ and $S'$), or we can write $S = \gen{g_1, \dots, g_{n-1}, g_n}$ and $S' = \gen{g_1, \dots, g_{n-1}, -g_n}$. Call an epistemic state of the form $T = \gen{g_1, \dots, g_{n-1}, h}$ with $h \notin S \cup S'$ a \emph{uniform coherent superposition} of $S$ and $S'$.

This definition generalizes the coherent superpositions for elementary systems found in \cite{spek} to composite systems. For any pair of distinct epistemic states, one being a rephasing of the other, there are four distinct coherent superpositions (using the same proof as for Lemma~\ref{supercount} above). The epistemic states thus obtained contain ontic states from both of the epistemic states in the superposition.

It was shown in Section~\ref{transsec} that for any pair of canonical generating sets, there is a permutation which maps one to the other. Take $S, S'$ and a coherent superposition $T$ written as above. Write a canonical generating set $g_1, \dotsc, g_n, h_1, \dotsc, h_n$. Fix $h_n = h$ ($h$ anticommutes with $g_n$ by the same argument used in the proof of Lemma~\ref{qubitsuper}). The other $h_k$ are irrelevant, provided they satisfy the requirements of a canonical generator set (such $h_k$ can always be found). Another canonical generating set is obtained by changing $g_n$ to $h$ and $h_n$ to $-g_n$. So there exists a permutation that sends $S$ to $T$ and $T$ to $S'$. This is considered in \cite{spek} to be indicative of a coherent superposition.

\section{Examples}\label{examples}
In this section the computational power of the notation is used to demonstrate some similarities of the toy theory to qubit stabilizers not identified in \cite{spek}. On the other hand, to show why not every phenomena involving qubit stabilizers is found in the toy theory, this section concludes with a discussion of the Mermin-Peres square.

\subsection{Counting states and transformations}
The number of pure stabilizer states on $n$ qubits is calculated in \cite{ag} to be
\begin{equation}
    2^n \prod_{k=0}^{n-1} (2^{n-k} + 1).
\end{equation}
Their argument applies equally to toy stabilizers and so the number of pure epistemic states for $n$ elementary systems in the toy theory is identical.

The number of ordered canonical generator sets, and hence the elements of the Clifford group (modulo global phases) is calculated in \cite{caves} to be
\begin{equation}
  2^{2n^2 + 3n}\prod_{k=1}^n (1 - 2^{-2k}).
\end{equation}
The same argument applies equally to toy stabilizers and so, recalling the link between transformations and canonical generator sets outlined in Section~\ref{transsec}, the number of valid reversible transformations in the toy theory is identical.

\subsection{Classical simulation}
  It was already noted in \cite{spek} that the ontic states of the toy theory can be tracked efficiently by a classical computer. Using the stabilizer notation a classical computer can furthermore efficiently track the epistemic state. Indeed the proof in \cite{ag} that the simulation of qubit stabilizer circuits is complete for the classical complexity class $\mathbf{\oplus L}$ (which is believed to be strictly weaker than $\mathbf{P}$, polynomial-time universal classical computation) can now easily be adapted to the toy theory.

\subsection{Graph states}
Let $G$ be a finite simple graph on $n$ vertices. We can associate $G$ with a pure state on $n$ elementary systems as follows. Begin with each system in the state $\gen{\X}$. Apply the two-system permutation that sends $\X_1 \to \X_1\Z_2$, $\X_2 \to \Z_1\X_2$, $\Z_1 \to \Z_1$ and $\Z_2 \to \Z_2$ (analogous to a controlled-Z gate) to each pair of systems connected by an edge on the graph (in any order). The toy stabilizer subgroup for the resulting state is generated by
\begin{equation}
  g_k = \X_k \prod_{l \in N(k)} \Z_l 
\end{equation}
  for $k \in \{ 1, \dotsc, n \}$ where $N(k)$ are the vertices connected to vertex $k$.

  Such states are analogous to graph states on qubits \cite{graph}, and share many qualitative features with them. For example, a $\Z_k$ commutes with all of the generators except for $g_k$. Hence a $\Z$ measurement on the $k$-th system will return $\pm 1$ with equal probability and generators for the new toy stabilizer can be found by replacing $g_k$ with $\pm \Z_k$. By multiplying the $g_l$ with $l \in N(k)$ by this we obtain another list of generators. If the value $-1$ is returned then apply the transformation that sends $\X \to -\X$ and $\Z \to \Z$ to the systems in $N(k)$. The new generators are now those of the graph state for $G$ with the vertex $k$ deleted, along with $\pm \Z_k$. Compare this to the qubit case, Proposition 1 of \cite{graph}. The effect of $\X$ and $\Y$ measurements is also analogous to the qubit case.

  \subsection{Non-example: Mermin-Peres square}
  Having seen all these examples it may tempting to conclude that any feature of the qubit stabilizer formalism has an analogue in the toy theory. This is not the case. The Mermin-Peres square \cite{mermin} of 9 Pauli observables 
  \begin{equation}
    \begin{matrix}
      X_1 & X_2 & X_1X_2 \\
      Y_2 & Y_1 & Y_1Y_2 \\
      X_1Y_2 & Y_1X_2 & Z_1Z_2
    \end{matrix}
    \label{mpsquare}
  \end{equation}
  has the property that every row and column is a set of commuting observables that multiply to give $I^{\otimes 2}$, except for the last row which gives $-I^{\otimes 2}$. Suppose we replace every observable in the square with some pre-determined value $\pm 1$, which is independent of how the observable is measured (non-contextual). To agree with quantum mechanical predictions each row and column of the square must multiply to $1$, except the last row which must multiply to $-1$. Since this is impossible, we conclude that in quantum mechanics observables do not have pre-determined non-contextual values.

  The square of toy observables corresponding to \eqref{mpsquare} has the property that every row and column is a set of commuting observables that multiply to give $\I^{\otimes 2}$. Hence the corresponding values must multiply to $1$, and this does not give rise to a contradiction. This is to be expected since in the toy theory all toy observables do indeed have a pre-determined non-contextual value, which could be calculated if the exact ontic state were known.
  
  \section{Summary}
  The similarities and differences between the qubit stabilizer formalism and Spekkens' toy theory can be summarised as follows. States, transformations, and measurements in their most compressed representation -- independent generators, canonical generator sets, and observables respectively -- appear to be identical in both theories and are manipulated according to the same procedures.
  
  Yet the underlying groups $G_n$ and $P_n$ are by no means identical, and so the results of ``decompression'' -- calculation of the full subgroup of a state, and the effect of a transformation on all observables -- will usually differ. This is the key difference between the two theories.

  With the notation in hand it becomes much easier to make calculations in the theory. This enables proofs of two important facts not shown in \cite{spek}: that the number of epistemic states on $n$ elementary systems is equal to the number of stabilizer states on $n$ qubits, and a similar result for transformations. Whilst dealing with just three elementary systems was previously very difficult, the stabilizer notation makes the consideration of states on a arbitrary number of elementary systems tractable.
\begin{acknowledgments}
I am grateful to my supervisors Terry Rudolph and Jonathan Barrett for many helpful discussions. I am particularly indebted to Terry for devising the representation of $G_n$ used here. I acknowledge financial support from the EPSRC.
\end{acknowledgments}
\appendix
\section{Proofs}\label{proofs}
\begin{proof}[Proof of Lemma~\ref{samegen}]
  The second statement is equivalent to ``$m(g_1), m(g_2), \dotsc, m(g_l)$ commute, have linearly independent check vectors, and each square to $I^{\otimes n}$'' \cite{caves}. By definition check vectors and commuting conditions are preserved by $m$. Furthermore ${m(g)}^2 = I^{\otimes n}$ since the only elements of $P_n$ that don't square to $I^{\otimes n}$ are those with phases $\alpha  \in \{i,-i\}$, which aren't in the image of $m$.

  Hence the second statement is equivalent to ``$g_1, g_2, \dotsc, g_l$ commute and have linearly independent check vectors''. The proof is completed by showing that this is equivalent to the first statement using a similar argument to the qubit case \cite{caves}. First note that $g_1, g_2, \dotsc, g_l$ commute if and only if $\gen{g_1, g_2, \dotsc, g_l}$ commute.
   
  Suppose $-\I^{\otimes n} \in \gen{g_1, g_2, \dotsc, g_l}$. Then the check vector of $-\I^{\otimes n}$ can be written is a linear combination of the check vectors of $g_1, g_2, \dotsc, g_l$. But the check vector of $-\I^{\otimes n}$ is $0$ and so the check vectors of $g_1, g_2, \dotsc, g_l$ are linearly dependent. Suppose that the first statement holds but the check vectors of $g_1, g_2, \dotsc, g_l$ are linearly dependent. Then the check vector of one of them, say of $g_1$, can be written as a linear combination of the others. That means that either $g_1$ or $-g_1$ can be written as a product of the others. The first possibility contradicts the assumption of independent generators, and since $g_1(-g_1) = -\I^{\otimes n}$ the second contradicts the assumption that the subgroup does not contain $-\I^{\otimes n}$.
\end{proof}
\begin{lemma}An epistemic state is permitted by the toy theory if and only if it corresponds to a toy stabilizer subgroup.\label{balancestates}\end{lemma}
\begin{proof}
In order to prove that the toy stabilizer subgroups correspond exactly to the valid epistemic states of Spekkens' toy theory, we need a precise formulation of the valid states for $n$ elementary systems. Define them inductively as follows. An \emph{epistemic state} is a subset of the ontic states (which we assign a uniform probability distribution over). For $n=1$ a \emph{valid epistemic state} consists of either two or four ontic states. Suppose the valid states are defined for $n \leq k$. Define a \emph{valid measurement} on $n \leq k$ elementary systems as the discovery of which member of a partition of the ontic states into valid epistemic states the ontic state lies in. A valid epistemic state on $n=k+1$ elementary systems is such that for any non-trivial partition of the elementary systems into two subsystems $A$ and $B$, any possible outcome of any valid measurement of $A$ (which is assumed not to affect the ontic state of $B$) results in a valid epistemic state for $B$. 

The definition of the ``knowledge balance principle'' in \cite{spek} is not as mathematical as the one above. Since the valid epistemic states of an elementary system are stated explicitly in \cite{spek}, the definitions certainly agree for $n=1$. For $n>1$ the inductive structure made explicit above is implicit in \cite{spek}: examples are given of particular epistemic states of a composite system, such that a measurement on one subsystem leads to an invalid epistemic state for the other. This is always taken as proof that the epistemic state of the composite system was invalid, the validity of the definitions for the smaller subsystems is never questioned.

Also implicit in \cite{spek} is that for any epistemic state on $n$ elementary systems, there must exist a canonical set such that the agent is certain about the answers to $l$ of the questions, and completely ignorant of the remaining $2n - l$.  Hence the number of ontic states compatible with the epistemic state must be $2^{(2n - l)}$. This implicit requirement rules out the epistemic state $1 \vee 2 \vee 3$.  Furthermore, by the knowledge balance principle we must have $l \leq n$. Except for when $n=1$, this requirement might appear to be missing from the above definition. But since the Lemma can be proven without it, and any state corresponding to a toy stabilizer subgroup certainly satisfies this requirement, it turns out that there is no need to include it in the definition.
  
  We now proceed with the proof using induction on $n$. The only toy stabilizer subgroups for $n=1$  are $\gen{}, \gen{\X}, \gen{-\X}, \gen{\Y}, \gen{-\Y}, \gen{\Z}$ and $\gen{-\Z}$ which correspond to the 7 valid epistemic states for an elementary system.

  Suppose the claim has been proved for all $n \leq k$ and consider an epistemic state of an $n=k+1$ system. For the ``if'' part let $S$ be a toy stabilizer subgroup. Consider a non-trivial partitioning into subsystems $A$ and $B$ of sizes $n_A$ and $n_B$. Since $n_A, n_B < k$ we can assume the lemma holds for each. Therefore some outcome of some valid measurement of $A$ corresponds to a valid epistemic state and hence a toy stabilizer subgroup for $A$, which we denote $N'$. Tensor the operators in $N'$ with an $\I^{\otimes n_B}$ to obtain the corresponding knowledge about the entire system, denoted $N$. It is now known that before the measurement the ontic state was in the support of $P_NP_S = P_{\gen{N \cup S}}$.
  
  $\gen{N \cup S}$ may not represent a valid epistemic state since we have not yet taken into account the disturbance to $A$ due to the measurement. But since system $B$ is not disturbed the new epistemic state for $B$ alone will be given by the subgroup $S_B$ of $\gen{N \cup S}$ that applies to $B$ alone. It is useful to note that since $G_n$ is abelian, $\gen{N \cup S} = \{ gh | g \in N, h \in S\}$.
  
  Since neither $N$ nor $S$ contain $-\I^{\otimes n}$, the only way $-\I^{\otimes n}$ could be in $\gen{N \cup S}$ is if there is some $g \in N$ with $-g \in S$. But then the measurement would never have returned the result it did. Hence $-\I^{\otimes n} \notin \gen{N \cup S}$.

  Let $g, h \in S_B$, and decompose $g=ab$, $h=cd$ with $a, c \in N$ and $b, d \in S$. We have that $a$ commutes with $c$ and $b$ commutes with $d$. Since $a, c \in N$ are results from a measurement of $A$ alone we can write $a = a_A \otimes \I^{\otimes n_B}$ and $c = c_A \otimes \I^{\otimes n_B}$ with $a_A, c_A \in G_{n_A}$. Since $g, h \in S_B$ apply to $B$ alone we must have $b = a_A \otimes b_B$ and $d = c_A \otimes d_B$ with $b_B, d_B \in G_{n_B}$. This demonstrates that $a$ commutes with $d$ and $b$ commutes with $c$ also. Therefore $g$ and $h$ commute. Hence $S_B$ is a toy stabilizer subgroup, which represents a valid epistemic state by the inductive hypothesis. Therefore $S$ represents a valid epistemic state.

  For the ``only if'' part consider some valid epistemic state. First I show that the probability of obtaining an outcome from the measurement of a toy observable is always either $0$, $\frac12$ or $1$. Let $g^+$ be the event of obtaining the $+1$ outcome from measuring some $g \in G_n$. Decompose $g = a \otimes b$ with $a \in G_1$ and $b \in G_{n-1}$. If $g \vec v = \vec v$ then, writing $\vec v = \vec v_a \otimes \vec v_b$ we must have $a \vec v_a = \vec v_a, b \vec v_b = \vec v_b$ or $a \vec v_a = -\vec v_a$ and $b \vec v_b = -\vec v_b$. Hence
  \begin{multline}
    P(g^+) = P(a^+ \cap b^+) + P(a^- \cap b^-) \\= P(a^+)P(b^+|a^+) + P(a^-)P(b^-|a^-).
  \end{multline}
  All of the probabilities in the last expression are $0$, $\frac12$ or $1$ by the inductive hypothesis and the fact that measuring $a$ will result in a valid epistemic state for the other subsystem. Hence the only way $P(g^+)$ could fail to be $0$, $\frac12$ or $1$ is if $P(a^+) = P(a^-) = \frac12$ and $(P(b^+|a^+), P(b^-|a^-)) \in \{(0, \frac12), (\frac12, 0), (\frac12, 1), (1, \frac12)\}$. But by the inductive hypothesis $P(b^+) \in \{0, \frac12, 1\}$, and we have
  \begin{equation}
    P(b^+) = P(b^+|a^+)P(a^+) + P(b^+|a^-)P(a^-),
  \end{equation}
  which would give a contradiction in each of those four cases.

  We now construct a list of independent generators for the toy stabilizer subgroup corresponding to the valid epistemic state as follows. Identify the first elementary system as subsystem $A$ and the rest as a subsystem $B$. Consider an $\I$ measurement on system $A$, which is valid by the inductive hypothesis. This is certain to return the outcome 1 and results in a valid epistemic state for $B$. By the inductive hypothesis this corresponds to a toy stabilizer subgroup $S_B$, which will have some list of independent generators. Tensor each element of this list with $\I$ to obtain the first generators for $S$.
  
  Next consider an $\gen{\X}$ versus $\gen{-\X}$ measurement on $A$ (i.e. a measurement of the $\X$ toy observable). If we know that only a single outcome, $\pm 1$ is possible then doing the measurement cannot give us any knowledge about system $B$ that we didn't already have. Therefore we just need to add $\pm \X_A \otimes \I^{\otimes(n-1)}$ to our generators for $S$ (it is clear that this commutes with, and has a linearly independent check vector from, the existing elements of the list).
  
  If either outcome is possible then they must result in valid epistemic states and hence toy stabilizer subgroups $S_B^{\pm}$. Suppose there is a $g \in S_B^+$ with neither $g$ nor $-g$ in $S_B^-$. Recall that the measurement of $A$ does not disturb $B$. Consider a $g$ measurement on system $B$ and denote the event of the $+1$ outcome $g^+$. Denote the event of a $\pm 1$ outcome of the $\X$ measurement on $A$ by $\X^{\pm}$. Then we have
  \begin{multline} P(g^+) = P(g^+|\X^+) \times P(\X^+) + P(g^+|\X^-) \times P(\X^-) \\= 1 \times \frac12 + \frac12 \times \frac12 = \frac34 \end{multline}
    contradicting the earlier proof that $P(g^+) \in \{0, \frac12, 1\}$.

    There is a similar contradiction if there is $g \in S_B^-$ with neither $g$ nor $-g$ in $S_B^+$. Hence (in the language of Section~\ref{mixsec}) $S_B^-$ is a rephasing of $S_B^+$. Therefore either $S_B^+ = S_B^-$ (in which case the measurement provides no information about $B$ and we do nothing), or we can write a list of independent generators for $S_B^\pm$ as $g_1, \dotsc, g_{l-1}, \pm g_l$. $g_1, \dotsc, g_{l-1}$ hold for either outcome so they must already be in $S$. Hence we just need to add $\X \otimes g_l$ to our generators for $S$. It is clear that the list still commutes and has linearly independent check vectors.

Now consider a $\Y$ measurement on $A$ and repeat the process. If there is a new $g_l$, it must anticommute with the old (if present), since otherwise we could measure two commuting toy observables and determine the value of both $\X$ and $\Y$ on system $A$. Hence the new generator $\Y \otimes g_l$ commutes with the one added in the previous step, and certainly commutes with all the rest and has a linearly independent check vector.

Finally repeat the process for a $\Z$ measurement on $A$, but don't add a generators with a check vector linearly dependent on the existing ones. Any such element will already be in the subgroup.

Consider the measurement of some toy observable $g$ with neither $g$ nor $-g$ in $S$. It must either return one outcome with certainty, or either outcome with equal probability. In the former case it will be possible to find some contradiction with the above procedure, i.e. to find a reason why $g$ or $-g$ would have been added to $S$. The latter case is exactly what is predicted by the toy stabilizer subgroup $S$.

We have now constructed a toy stabilizer subgroup $S$ whose corresponding epistemic state makes identical predictions about the expectation values of all toy observables. It remains to show that that there is only one epistemic state with this property and hence $S$ represents exactly the epistemic state we began with. This is analogous to the fact that knowing the expectation values of every Pauli observable uniquely determines a quantum state, and the proof is similar. We note that the elements $G_n^+$ of $G_n$ that have $\alpha = 1$ are a basis for the real vector space of real diagonal $4^n \times 4^n$ matrices. This is true because there are $4^n$ of them, and they are orthogonal under the trace inner product and hence linearly independent. Hence any probability distribution over the ontic states, written as a diagonal real $4^n \times 4^n$ matrix, is a linear combination of the $G_n^+$ where the coefficients are proportional to the expectation values of those observables.
\end{proof}
\begin{proof}[Proof of Lemma~\ref{onetwomeas}]
    For one system this can be seen by inspection. 
    
    Consider a four-outcome valid measurement on two systems and represent the outcomes, which are valid epistemic states, by toy stabilizer subgroups $S_1, S_2, S_3, S_4$. I claim there are three toy observable $g, h^\pm$ with $\{g, h^+\}$, $\{g, -h^+\}$, $\{-g, h^-\}$, $\{-g, -h^-\}$ each contained in one of the $S_k$. Hence the measurement can be implemented by measuring $g$, and then, based on the outcome $\pm 1$, measuring $h^\pm$. For example if $S_1 = \gen{\X_1\X_2, \Y_1\Y_2}$, $S_2 = \gen{-\X_1\X_2,-\Y_1\Y_2}$, $S_3 = \gen{\Z_1, -\Z_2}$ and $S_4 = \gen{-\Z_1, \Z_2}$ we can take $g = \Z_1\Z_2$, $h^+ = \X_1\X_2$ and $h^- = \Z_1$.

    First note that since the $S_k$ are disjoint and cover all the ontic states
    \begin{equation}
      P_{S_1} + P_{S_2} + P_{S_3} + P_{S_4} = \I^{\otimes 2}.
      \label{sumident}
    \end{equation}
    Denote $S_1 = \{ \I^{\otimes 2}, a, b, ab \}$, where $a$ and $b$ are independent generators of $S_1$. By \eqref{sumident}  $-a$, $-b$ and $-ab$ must appear in the other $S_k$. Suppose without loss of generality that $-a$ appears in $S_2$ and denote $S_2 = \{ \I^{\otimes 2}, -a, c, -ac \}$. If $c = b$ then by \eqref{sumident} we have $-b \in S_3, S_4$ and so we can take $g = b$. Similarly if $c = -b$ take $g = ab$, if $c = ab$ take $g = ab$ and if $c = -ab$ take $g = b$. Suppose $c$ is otherwise. Since $-b$, $-c$ and $-ac$ must be in $S_3$ or $S_4$, two of them must be in the same subgroup. If $-c$ and $-ac$ are in the same subgroup then $a$ is also in it and so by \eqref{sumident} we can take $g =a$. If $-b$ and $-c$ are in the same subgroup then they must commute. But then $\{ a, b, c \}$ would be independent generators of a toy stabilizer subgroup, even though the maximum length of such a list is $n=2$. Similarly if $-b$ and $-ac$ being in the same subgroup creates a contradiction using $\{ a, b, ac \}$.

    Once we have found $g$, use \eqref{sumident} to verify that $g$ must be in two of the subgroups and $-g$ must be in the other two. Take the two subgroups with $g$ in. To ensure that they represent disjoint states there must be some $h^+$ in one with $-h^+$ in the other. Define $h^-$ similarly.

    The proofs for a valid measurement with two or three outcomes are simpler, but use similar observations.
  \end{proof}
\begin{proof}[Proof of Lemma~\ref{qubitsuper}]
  We begin with the ``if'' part. For any $\theta \in \mathbb{R}$ and $l \in \{ 1, \dotsc, n-1\}$ we have $g_l \ket{\psi} = \ket{\psi}$ and $g_l\ket{\psi'} = \ket{\psi'}$, hence
  \begin{equation}
    g_l \frac1{\sqrt 2}(\ket{\psi} + e^{i\theta}\ket{\psi'}) = \frac1{\sqrt 2}(\ket{\psi} + e^{i\theta}\ket{\psi'}).
  \end{equation}

  Next calculate the effect of $h$ on $\ket{\psi}$. Since $T$ is a stabilizer subgroup $h$ must commute with $g_1, \dotsc, g_{n-1}$. But it must anticommute with $g_n$ since otherwise $g_1, \dotsc, g_{n_1}, g_n, h$ would be a list of $n+1$ independent generators of a stabilizer subgroup, contradicting the maximum length of such a list being $n$. Therefore under conjugation by $h$, $S$ is mapped to $S'$. Hence $h\ket{\psi} = e^{i\theta}\ket{\psi'}$ for some $\theta \in \mathbb{R}$. Since $h$ is Hermitian and unitary we also have $he^{i\theta}\ket{\psi'} = \ket{\psi}$. Therefore $h\frac1{\sqrt 2}(\ket\psi + e^{i\theta}\ket{\psi'}) = \frac1{\sqrt2}(\ket\psi + e^{i\theta}\ket{\psi'})$, and we have that $T$ stabilizes this state.
  
  Now for the ``only if'' part. Suppose $S'$ is not a rephasing of $S$. Then there exists some $g \in S$ with $\pm g \notin S'$. Denote $\ket{\phi(\theta)} = \frac1{\sqrt 2}(\ket{\psi} + e^{i\theta}\ket{\psi'}$ and calculate
  \begin{multline}
    \Braket{\phi(\theta) | g | \phi(\theta)} = \frac12(\Braket{\psi|g|\psi} + \Braket{\psi'|g|\psi'} \\+ e^{i\theta}\Braket{\psi|g|\psi'} + e^{-i\theta}\Braket{\psi'|g|\psi}) = \frac12,
  \end{multline}
  which contradicts the expectation value of any Pauli observable in a stabilizer state being $-1, 0$ or $1$. 
  
  So $S'$ is a rephasing of $S$. Since $S' \neq S$ we can therefore write $S = \gen{g_1, \dotsc, g_{n-1}, g_n}$ and $S = \gen{g_1, \dotsc, g_{n-1}, -g_n}$. It is easy to check that for $l \in \{1, \dotsc, n-1\}$ we have $g_l \ket{\phi(\theta)} = \ket{\phi(\theta)}$ and hence $g_l \in T$. Therefore we can take $g_1, \dotsc, g_{n-1}$ as the first $n-1$ independent generators of $T$. Write the final independent generator as $h$. Suppose $h \in S$. Since $h \notin \gen{g_1, \dotsc, g_{n-1}}$, and $S'$ is a rephasing of $S$, we must have $-h \in S$. But then $h\ket{\phi(\theta)} = \ket{\phi(\theta + \pi)}$ contradicting $h$ being a stabilier of $\ket{\phi(\theta)}$. Similarly if $h \in S'$. Therefore $h \notin S \cup S'$.
\end{proof}

\begin{proof}[Proof of Lemma~\ref{supercount}] We closely follow the proof of Proposition 1 in \cite{ag}. The number of distinct stabilizer subgroups of the required form is $G/A$ where $G$ is the number of choices for the final generator $h$, and $A$ is the number of choices of $h$ giving rise to the same subgroup.

  First we calculate $G$. Ignoring phases there are $4^n$ choices in $G_n$. But $h$ must commute with $g_1, \dotsc, g_{n-1}$ which gives $4^n / 2^{n-1}$ options. Also, it must not be in the subgroup generated by $g_1, \dotsc, g_{n-1}$, which gives $4^n / 2^{n-1} - 2^{n-1}$. Finally, it must not be in $S$ or $S'$, i.e. it cannot be $g_n$ multiplied by something generated by $g_1, \dotsc, g_{n-1}$. There are two possible overall phases $\pm 1$. This gives $G =  2(4^n / 2^{n-1} - 2^{n-1} - 2^{n-1}) = 2(4 \times 2^{n-1} - 2 \times 2^{n-1}) = 2^{n+1}$.

  Next we calculate $A$. Fix a stabilizer subgroup. Then for $h$ we can choose any of the elements of the stabilizer subgroup not generated by $g_1, \dotsc, g_{n-1}$. This gives $A = 2^n - 2^{n-1} = 2^{n-1}$.

  Finally we calculate $G/A = 4$ as required.
\end{proof}
  \section{Relation to Spekkens' new formulation}\label{newformulation}
  Spekkens' has previously outlined \cite{spektalk} a new formulation of the toy theory which is very closely related to the stabilizer notation presented here. The new formulation is based around ``linear functionals'' or ``canonical variables''. These can be put in 1-to-1 correspondence with the elements of $G_n$ with $\alpha = 1$. For an elementary system, the canonical variables are $0$, $X$, $P$, and $X+P$ which may be taken to correspond to $\I$, $\X$, $\Z$ and $\Y$ respectively. The correspondence for composite systems can then be built up from this, so that $X_1 + X_3 + P_3$ corresponds to $\X\otimes\I\otimes\Y$ and so on. The toy stabilizer notation combines canonical variables with their values, for example having $\X\otimes\Z$ in a toy stabilizer subgroup represents knowledge that $X_1 + P_2 = 0$ whereas $-\X\otimes\Z$ represents the knowledge $X_1 + P_2 = 1$.

  The new formulation's ``poisson bracket'' condition for jointly-knowable variables is exactly the usual condition on the check vectors of commuting observables. 

  Measurements in the new formulation determine the value of some set of jointly-knowable variables, or equivalently a set of commuting toy observables. Therefore Lemmas~\ref{onetwomeas} and \ref{threemeas} show that for one or two systems the notion of measurement is identical in both formulations, whereas for three or more systems not all of the measurements in the original formulation are included in the new formulation.

  The new formulation can be compared to qubit quantum mechanics by using the discrete Wigner representation\cite{discretewig,gross} of the latter. From this perspective the difference between the two theories comes from the fact the discrete Wigner function is sometimes negative, whereas the toy theory only uses positive probabilities.

\bibliography{pseudostab}
\end{document}